\documentclass[letterpaper, 10 pt, conference]{ieeeconf}
\IEEEoverridecommandlockouts                             
\usepackage{amsfonts}
\usepackage{amsmath} 
\usepackage{amssymb}  

\usepackage{float}

\usepackage{hhline}
\usepackage{bm}
\usepackage{makecell}
\newcolumntype{C}[1]{>{\centering\let\newline\\\arraybackslash\hspace{0pt}}m{#1}}

\usepackage{amsthm}

\theoremstyle{definition}
\newtheorem{proposition}{Proposition}[section]

\usepackage{hyperref}
\hypersetup{
    colorlinks=true,
    linkcolor=blue,
    filecolor=magenta,      
    urlcolor=cyan,
}

\usepackage{tikz}
\usetikzlibrary{calc,patterns,decorations.pathmorphing,decorations.markings}
\usetikzlibrary{shapes,arrows}
\usepackage{verbatim}

\tikzstyle{block} = [draw, fill=blue!20, rectangle, 
    minimum height=3em, minimum width=6em]
\tikzstyle{sum} = [draw, fill=blue!20, circle, node distance=1cm]
\tikzstyle{input} = [coordinate]
\tikzstyle{output} = [coordinate]
\tikzstyle{pinstyle} = [pin edge={to-,thin,black}]

\usetikzlibrary{positioning}
\usetikzlibrary{math}

\usepackage{tikz}
\usetikzlibrary{shapes,arrows,calc,positioning}
\tikzstyle{bigblock} = [draw, fill=blue!20, rectangle, 
    minimum height=6em, minimum width=8em]
\tikzstyle{medblock} = [draw, fill=blue!20, rectangle, 
    minimum height=4em, minimum width=4em]    
\tikzstyle{mux} = [draw, fill=black!20, rectangle, 
    minimum height=5em, minimum width=0.1em]    
\tikzstyle{smallblock} = [draw, fill=blue!20, rectangle, 
    minimum height=3em, minimum width=4em]
\tikzstyle{sum} = [draw, fill=blue!20, circle, node distance=1cm]
\tikzstyle{signal} = [coordinate]
\tikzstyle{pinstyle} = [pin edge={to-,thin,black}]
\tikzstyle{block} = [draw, fill=blue!20, rectangle, 
    minimum height=3em, minimum width=6em]
\tikzstyle{blockS} = [draw, fill=blue!20, rectangle, 
    minimum height=3em, minimum width=4em]    
\tikzstyle{input} = [coordinate]
\tikzstyle{output} = [coordinate]
\usetikzlibrary{matrix}

\newcommand{\bc}{\begin{center}}
\newcommand{\ec}{\end{center}}
\newcommand{\benum}{\begin{enumerate}}
\newcommand{\eenum}{\end{enumerate}}
\newcommand{\nn}{\nonumber}
\newcommand{\matl}{\left[ \begin{array}}
\newcommand{\matr}{\end{array} \right]}
\newcommand{\matls}{\left[ \begin{smallmatrix}}
\newcommand{\matrs}{\end{smallmatrix} \right]}
\newcommand{\isdef}{\stackrel{\triangle}{=}}

\newcommand{\vect}[1]{\overset{\rightharpoonup}{#1}}

\newcommand{\rmE}{{\rm E}}

\newcommand{\rmG}{{\rm G}}

\newcommand{\rmI}{{\rm I}}

\newcommand{\rmT}{{\rm T}}

\newcommand{\rmc}{{\rm c}}
\newcommand{\rmd}{{\rm d}}

\newcommand{\rmi}{{\rm i}}

\newcommand{\rmm}{{\rm m}}

\newcommand{\rmp}{{\rm p}}

\newcommand{\rms}{{\rm s}}

\newcommand{\BBR}{{\mathbb R}}

\newcommand{\shiftq}{{\textbf{\textrm{q}}}}

\newcommand{\resolvedinFrame}[1]{{\big|_{\rm #1}}}

\newcommand{\ihat}{ {\hat \imath}}
\newcommand{\jhat}{ {\hat \jmath}}
\newcommand{\khat}{ {\hat k}}

\title{An A Quadcopter Autopilot Based on an Adaptive Digital PID Controller}

\title{Adaptive Digital PID Control of a Quadcopter}

\title{Retrospective-Cost-Based Adaptive Digital PID Control of a Quadcopter}

\title{A Retrospective-Cost-Based Adaptive Digital PID   Quadcopter Autopilot}

\title{Adaptive Digital PID Control of a Quadcopter with Unknown Dynamics}

\title{One-Shot Learning for a Quadcopter Autopilot}

\title{An adaptive digital autopilot for Multicopters}

\title{\LARGE \bf Experimental Implementation of an Adaptive Digital Autopilot with Applications }

\title{\LARGE \bf An Adaptive Digital Autopilot with Applications \\ for Fixed-wing Aircraft Control }

\title{\LARGE \bf An Adaptive Digital Autopilot\\ for Fixed-Wing Aircraft with Actuator Faults}

\author{
    Joonghyun Lee,
    John Spencer,
    Juan Augusto Paredes, 
    Sai Ravela, 
    Dennis S. Bernstein,
    Ankit Goel%
\thanks{This research was supported in part by the Office of Naval Research under grant N00014-19-1-2273.}
\thanks{Joonghyun Lee, John Spencer, Juan Augusto Paredes, and Dennis S. Bernstein are with the Department of Aerospace Engineering, University of Michigan, Ann Arbor, MI 48109.
{\tt\small joonghle, spjohn, jparedes,}
{\tt \small dsbaero@umich.edu}
}
\thanks{Sai Ravela is with the Department of Earth, Atmospheric, and Planetary Sciences, MIT, Cambridge, MA 02139.
{\tt\small ravela@mit.edu}}
\thanks{Ankit Goel is with the Department of Mechanical Engineering, University of Maryland, Baltimore County, MD 21250.
{\tt \small ankgoel@umbc.edu}}
}

\date{}

\begin{document}

\maketitle

\begin{abstract}
    This paper develops an adaptive digital autopilot for a fixed-wing aircraft and compares its performance with a fixed-gain autopilot.
    The adaptive digital autopilot is constructed by augmenting the autopilot architecture implemented in PX4 flight stack with  adaptive digital control laws that are updated using the retrospective cost adaptive control algorithm.  
    In order to investigate the performance of the adaptive digital autopilot, the default gains of the fixed-gain autopilot are scaled down to degrade its performance.
    This scenario provides a venue for determining the ability of the adaptive digital autopilot to compensate for the detuned fixed-gain autopilot.
    Next, the performance of the adaptive autopilot is examined under failure conditions by simulating a scenario where one of the control surfaces is assumed to be stuck at an unknown angular position. 
    The adaptive digital autopilot is tested in simulation, and the resulting performance improvements are examined.  
    
\end{abstract}

\section{Introduction}

Autonomous flight control of an aircraft under rapidly changing conditions depends on an effective autopilot that can control the aircraft in uncertain environments and without detailed models. 
The autopilot for fixed-wing aircrafts typically consists of a set of trim commands along with low-level controllers to track intermediate commands.
The trim conditions for an aircraft can be computed by solving nonlinear algebraic equations for trim equilibria \cite{mcclamroch2011steady}, but a detailed knowledge of aircraft aerodynamics is required.
Moreover, for low-cost aircrafts that are usually repaired or modified onsite, the true aerodynamic properties may be different from nominal aerodynamics. 
Consequently, a fixed-gain autopilot may not be able to maintain performance in rapidly changing environment or under failure conditions such as chipped wings or stuck actuators. 
In this scenario, an adaptive autopilot may be able to compensate for the lost performance by updating the autopilot gains accordingly. 
With these motivations in mind, this paper explores the use of an in situ learning technique to modify the autopilot during the flight. 


Several adaptive control techniques have been investigated for the problem of fixed-wing aircraft control
\cite{nguyen2006dynamics}. 
A sliding mode fault-tolerant tracking control scheme was used for control of a fixed-wing UAV under actuator saturation and state constraints in \cite{9262225, 9476716}. 
A backstepping algorithm was used  in \cite{hirano2019controller} to design a nonlinear flight controller for a fixed-wing UAV with thrust vectoring. 
An MRAC-based control technique was used to augment the control system to improve the dynamic performance of a fixed-wing aircraft in \cite{9189264}. 
However, these techniques rely on the availability of a sufficiently detailed model for the control system synthesis. 
In contrast, this paper uses the the retrospective cost adaptive control (RCAC) algorithm to learn the autopilot from the measured data in situ. 
RCAC is a digital adaptive control technique that is applicable to stabilization, command-following, and disturbance-rejection problems. 
Instead of relying on a model of the system, RCAC uses the past measured data and past applied input to recursively optimize the controller gains.
RCAC is described in detail in \cite{rahmanCSM2017} and its extension to digital PID control is given in \cite{rezaPID}.
The application of the RCAC algorithm for a multicopter autopilot are described in \cite{goel_adaptive_pid_2021,quadtuner2021}.




The contribution of this paper is the development of an adaptive digital autopilot, and comparison of its performance with a well-tuned fixed-gain autopilot under nominal conditions, rapidly changing conditions, and actuator failure. 
In particular, this paper presents the potential advantages of an adaptive autopilot by investigating two scenarios. 
In the first scenario, a well-tuned fixed-gain controller is detuned by scaling all of the gains by a small factor, and it is shown that the adaptive autopilot is able to compensate for the detuned gains by learning the necessary gains.
In the second scenario, the aircraft is simulated with a frozen aileron, thus emulating an actuator failure condition, and it is shown that the adaptive autopilot  improves the trajectory-tracking performance.
%

%

The paper is organized as follows:
Section \ref{sec:notation} defines the notation used in this paper, %
Section \ref{sec:PX4_autopilot} reviews the autopilot architecture implemented in the PX4 flight stack,
Section \ref{sec:RCAC} reviews the retrospective cost adaptive algorithm,
Section \ref{sec:adaptiveAugmentation} presents the adaptive augmentation of autopilot, 
Section \ref{sec:flight_tests} presents the simulation tests,
and finally, 
Section \ref{sec:conclusions} concludes the paper with a summary and future research directions.

\section{Notation}
\label{sec:notation}
Let $\rm F_E$ denote an Earth-fixed frame such that $\khat_\rmE$ is aligned with gravity $\vect g.$
Let $\rm F_{AC}$ denote an aircraft-fixed frame such that $\ihat_{\rm AC}$ is aligned with the fuselage, 
$\jhat_{\rm AC}$ is along the wing, and 
$\khat_{\rm AC}$ is chosen to complete the right-handed frame. 
Note that $\khat_{\rm AC}$ points vertically down. 
Next, let $\rmc$ denote the center of mass of the aircraft, and let $w$ be an point fixed on Earth.
The coordinates of the aircraft relative to $w$ in the Earth frame are denoted by $r \isdef \vect r_{\rmc/w}\resolvedinFrame{E} \in \BBR^3.$
The velocity of the aircraft relative to $w$ in the Earth frame is  $ v\isdef  \vect v_{\rmc/w/\rmE}\resolvedinFrame{E} \in \BBR^3.$ 
%
Let $\Theta, \Psi,$ and $\Phi$ denote the pitch, azimuth, and the roll angle of the aircraft. 
The angular velocity of $\rm F_{AC}$ relative to  $\rm F_E$ in the aircraft-fixed frame is given by $\omega \isdef \vect \omega_{{\rm AC} / \rmE}\resolvedinFrame{\rm AC} \in \BBR^3.$
The angular acceleration of $\rm F_{AC}$ relative to  $\rm F_E$ in the aircraft-fixed frame is given by $\alpha \isdef \vect \alpha_{{\rm AC} / \rmE}\resolvedinFrame{\rm AC} \in \BBR^3.$
%
The measurement of the variable $x$ is denoted by $x_\rmm$, and the setpoint for the variable $x$ is denoted by $x_\rms.$
Finally, let $e_3 \isdef \matl{c c c} 0 & 0 & 1 \matr^{\rmT}.$

\section{PX4 Autopilot}
\label{sec:PX4_autopilot}
In this section, the control system implemented in the PX4 autopilot for fixed-wing aircraft is described. %
%
The control system consists of a mission planner and two cascaded controllers in nested loops as shown in Figure  \ref{fig:PX4_autopilot_nested_loop}.
The mission planner generates position setpoints based on the user-defined waypoints. 

\begin{figure}[h]
    \centering
    \resizebox{\columnwidth}{!}{
    \begin{tikzpicture}[auto, node distance=2cm,>=latex',text centered]
    
        \node [smallblock, minimum height=3em, text width=1.6cm] (Mission) {\small Mission Planner};
        \node [smallblock, minimum height=3em, right = 5em of Mission, text width=1.6cm] (Pos_Cont) {\small Position Controller};
        \node [smallblock, minimum height=3em, below right = 1.75em and -2.25 em of Pos_Cont,text width=1.6cm] (Att_Cont) {\small Attitude Controller};
        \node [smallblock, minimum height=3em, text width=1.6cm, minimum width = 5.5em,  right = 4em of Pos_Cont] (FWAircraft) {\small Fixed-Wing Aircraft};
        
        \draw [->] (Mission) -- node[above, xshift = -0.05 em, yshift = 0.1em]{ {}$r_\rms, V_{\rmT, \rms}$} (Pos_Cont);
        \draw [->] (Pos_Cont.-10) -| ([xshift = 0.75em, yshift = -2em]Pos_Cont.-10) -| node  [xshift = 3.15em, yshift = 1em] {\small {}$\begin{array}{c} \Phi_{\rms} , \\ \Theta_{\rms} \end{array}$} ([xshift = -0.75em]Att_Cont.west) -- (Att_Cont.west);
        \draw [->] (Pos_Cont.10) -- node [above, xshift = 0.05em, yshift = 0.1 em] { \small{}$T _\rms$} (FWAircraft.170);
        \draw [->] (Att_Cont.0) -- +(0.15,0) |- node [below, xshift = 0.7 em, yshift = -2 em]{\small {}$\alpha_{\rms}$}(FWAircraft.190);
        \draw [->] (FWAircraft.0) -- +(.4,0) |- node[below,xshift = -6em, yshift = 0.05em]{ \small {} $r_{\rmm}, V_{\rmT}, V_{\rmG}$} ([xshift = -1.5em, yshift = -7.8 em]Pos_Cont.south) --  ([xshift = -1.5em]Pos_Cont.south);
        \draw [->] (FWAircraft.350) -- +(0.2,0) |- node[below,xshift = -5em]{ \small {} $ \Phi_{\rmm}, \Theta_{\rmm}, V_{\rmT}, V_{\rmI} , \omega_{\rmm}$} ([yshift = -0.75 em]Att_Cont.south) -- (Att_Cont.south);
         \draw [->] (FWAircraft.10) -- +(.6,0) |- node[below,xshift = -6em, yshift = -0.05em]{ \small {} $V_{\rmG}$} ([yshift = -10 em]Mission.south) --  (Mission.south);
    \end{tikzpicture}
    }
    \vspace{-2em}\caption{\footnotesize PX4 autopilot architecture. }
    \label{fig:PX4_autopilot_nested_loop}
\end{figure}
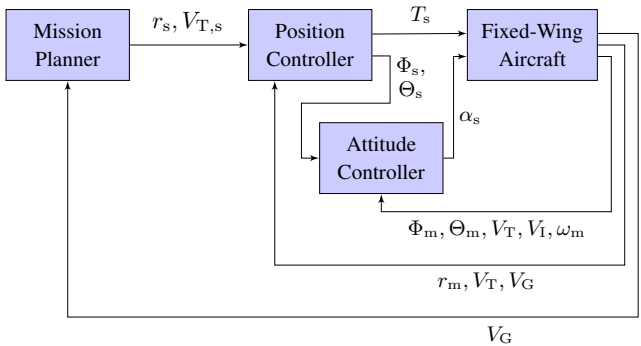
%

%
%
%
The outer loop, also called the position controller, consists of two decoupled controllers for the longitudinal and the lateral motion of the aircraft, and is shown in Figure \ref{fig:PX4_autopilot_outer_loop}. 
The longitudinal controller is based on the total energy control system (TECS) described in \cite{bruce1986,faleiro1999,lambregts2013,Argyle2016} and the lateral controller is based on the guidance law described in \cite{park2004}.
%
The input to the position controller are the true airspeed setpoint $V_{\rmT, \rms},$ the position setpoint $r_\rms,$ the true airspeed $V_{\rmT},$ the position measurement $r_\rmm,$ and the ground velocity $V_{\rmG}.$ 
Note that the TECS input include the altitude setpoint $h_\rms \isdef e_3^\rmT r_\rms,$ and the altitude measurement $h _\rmm \isdef e_3^\rmT r_\rmm.$
The longitudinal controller generates the thrust and the pitch setpoint, and the lateral controller generates the roll setpoint. 
%
%
The output of the position controller is thus the thrust setpoint $T_\rms$ and the attitude setpoint ($\Phi_\rms, \Theta_\rms,$ since $\Psi_\rms$ is not required by the attitude controller).

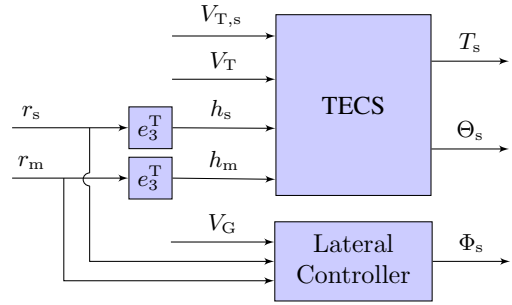
\begin{figure}[h]
    \centering
    \resizebox{0.8\columnwidth}{!}{
    \begin{tikzpicture}[auto, node distance=2cm,>=latex']
        \node[smallblock, minimum width = 6em, minimum height = 7 em] (TECS) {TECS};
        %
        %
        \node[smallblock, below = 4.5em of TECS.center, minimum width = 2.5em, minimum height = 3 em] (L1C) {$\begin{array}{c}{\rm Lateral} \\ {\rm Controller}\end{array}$};
        \node[smallblock, below left = 0.05em and 7em of TECS.center, minimum width = 1em, minimum height = 1 em] (r2h1) {\small $e_3^\rmT$};
        \node[smallblock, below left = 2em and 7em of TECS.center, minimum width = 1em, minimum height = 1 em] (r2h2) {\small $e_3^\rmT$};
        \draw[->] (r2h1.east) -- node[above, xshift = -0.1em]{\small$h_{\rms}$} ([yshift = -0.9em]TECS.west);
        \draw[->] (r2h2.east) -- node[above]{\small$h_{\rmm}$} ([yshift = -2.875em]TECS.west);
        \draw[->] ([xshift = -4.5em]r2h1.west) -- node[above, xshift = -1.5em]{\small$r_\rms$} (r2h1.west);
        \draw[->] ([xshift = -4.5em]r2h2.west) -- node[above, xshift = -1.5em]{\small$r_\rmm$} (r2h2.west);
        \draw[->] ([xshift = -4em, yshift = 2.675em]TECS.west) -- node[above, xshift = -0.1em]{\small$V_{\rmT, \rms}$} ([yshift = 2.675em]TECS.west);
        \draw[->] ([xshift = -4em, yshift = 1em]TECS.west) -- node[above]{\small$V_{\rmT}$} ([yshift = 1em]TECS.west);
        \draw[->] ([xshift = -4em, yshift = 0.75em]L1C.west) -- node[above]{\small$V_{\rmG}$} ([yshift = 0.75em]L1C.west);
        \draw[-] ([xshift = -1.5em]r2h1.west)--([xshift = -1.5em, yshift = -1.75em]r2h1.west);
        \draw ([xshift = -1.5em, yshift = -2.25 em]r2h1.west) arc (270:90:0.25em);
        \draw[->] ([xshift = -1.5em, yshift = -2.25em]r2h1.west) |- (L1C.west);
        \draw[->] ([xshift = -2.5em]r2h2.west) |- ([yshift = -0.75em]L1C.west);
        \draw[->] ([yshift = 1.7em]TECS.east) -- node[above]{\small $T_{\rms}$} ([xshift = 3em, yshift = 1.7em]TECS.east);
        \draw[->] ([yshift = -1.7em]TECS.east) --  node[above]{\small $\Theta_{\rms}$} ([xshift = 3em, yshift = -1.7em]TECS.east);
        \draw[->] (L1C.east) --  node[above]{\small $\Phi_{\rms}$} ([xshift = 3em]L1C.east);
    \end{tikzpicture}
    }
    \caption{\footnotesize PX4 autopilot position controller. }
    \label{fig:PX4_autopilot_outer_loop}
\end{figure}

The inner loop, also called the attitude controller, consists of two cascaded controllers, and is shown in Figure \ref{fig:PX4_autopilot_inner_loop}. 
The first controller uses the pitch and roll errors and a proportional control law to generate the pitch-rate and roll-rate setpoints. 
In particular, the pitch-rate setpoint $\dot{\Theta}_{\rms}$ and the roll-rate setpoint $\dot{\Phi}_{\rms} $ are given by
\begin{align}
    \dot{\Theta}_{\rms} 
        &=
            k_\theta(\Theta_{\rms} - \Theta_{\rmm})
    \label{eq:pitch_rate_P}
        , \\        
    \dot{\Phi}_{\rms} 
        &=
            k_\phi (\Phi_{\rms} - \Phi_{\rmm}), 
    \label{eq:Roll_rate_P}
\end{align}
where $k_{\theta}, k_{\phi}$ are the scalar gains.
The yaw-rate is computed algebraically by 
\begin{align}
    \dot{\Psi}_{\rms} 
        = \frac{g \tan\Phi_{\rms} \cos\Theta_{\rms}}{V_{\rmT}}
    \label{eq:coordinated_turn}
\end{align}
to ensure coordinated turn. 
The body-fixed angular-velocity setpoint $\omega_\rms
$ is given by
\begin{align}
    \omega_\rms
        &= 
            S(\Theta_{\rmm}, \Phi_{\rmm})
            \begin{bmatrix} 
                \dot{\Phi}_{\rms} \\ \dot{\Theta}_{\rms} \\ \dot{\Psi}_{\rms} 
            \end{bmatrix},
    \label{eq:e2omega}
\end{align}
where
\begin{align}
    S(\theta, \Phi)
        \isdef
    \begin{bmatrix} 
        1 & 0 & \sin \Theta \\
        0 & \cos \Phi & \sin \Phi \cos \Theta \\
        0 & -\sin\Phi & \cos\Phi \cos\Theta
    \end{bmatrix}.
\end{align}

Next, a feedforward and a PI control law generates the anglular-acceleration setpoint $\alpha_s
$.
In particular, $\alpha_s $ is given by
\begin{align}
    \alpha_s
        &=
            \frac{V_{\rmT, 0}}{V_{\rmT}} G_{\omega,{\rm ff}} 
            \omega_\rms+
            \left(\frac{V_{\rmI,0}}{V_{\rmI}}\right)^2 G_{\omega,{\rm PI}}(\shiftq) 
            \left( \omega_\rms - \omega_\rmm
            \right), \label{eq:alpha_s}
\end{align}
where $G_{\omega,{\rm ff}} = k_{\omega, \rm ff} $ is a proportional control law, 
$G_{\omega,{\rm PI}}(\shiftq)  = k_{\omega, {\rm P}} + \dfrac{k_{\omega, {\rm I}}}{\textbf{q}-1}$ is a PI control law,
$V_{\rmI}$ is the indicated airspeed,
and $V_{\rmT,0}$ and $V_{\rmI,0}$ are the true airspeed and the indicated airspeed at trim conditions respectively, which are aircraft parameters.
%
Note that 
$\shiftq$ is the forward-shift operator, 
$k_{\omega, \rm ff}, $
$k_{\omega, \rm P},$ and
$k_{\omega, \rm I}$ are $3\times 3$ diagonal matrices, and are thus parameterized by 9 scalar gains. 
Finally, using the angular-acceleration setpoint, the actuator deflections are computed using control allocation methods.

\begin{figure}[h]
    \centering
    \resizebox{\columnwidth}{!}{
    \begin{tikzpicture}[auto, node distance=2cm,>=latex']
    \node[smallblock, minimum width = 2.5em, minimum height = 4 em] (e2q) { \eqref{eq:e2omega}};
    \node[smallblock, above left = -0.5em and 4em of e2q.center, minimum width = 2.5em, minimum height = 2.5em] (G_E) {\eqref{eq:pitch_rate_P}, \eqref{eq:Roll_rate_P}};
    \node[smallblock, below left = 1.5em and 4em of e2q.center, minimum width = 2.5em, minimum height = 2.5em] (tc) {\eqref{eq:coordinated_turn}};
    \node[sum, left = 2.5em of G_E.center] (suml1){};
    \node[draw = none] at (suml1.center) {$+$};
    %
    %
    \node[sum, right = 2.5 em of e2q] (sumr1){};
    \node[draw = none] at (sumr1.center) {$+$};
    \node[smallblock, right = 1em of sumr1, minimum width = 2.5em, minimum height = 1.75 em] (G_PI) {$G_{\omega, {\rm PI}}$};
    \node[smallblock, above = 1em of G_PI, minimum width = 2.5em, minimum height = 1.75 em] (G_FF) {$G_{\omega, {\rm FF}}$};
    \node[smallblock, right = 1em of G_PI, minimum width = 2.5em, minimum height = 1.75 em] (sc) {\eqref{eq:alpha_s}};
    \draw [->] (sc.east) -- ([xshift = 1.5em]sc.east) node [above, xshift=-0.4 em] { {}$\alpha_s$};
    %
    \draw[->]([yshift = -0.25em]G_E.east) -- node[above]{ {} $\begin{bmatrix} \dot{\Phi}_{\rms} \\ \dot{\Theta}_{\rms} \end{bmatrix}$}([yshift = 0.5em]e2q.west);
    \draw[->](tc.east)-|([xshift = -2em, yshift = -1.25em]e2q.west)-- node[above, xshift = -0.35 em]{ {}$\dot{\Psi}_{\rms}$} ([yshift = -1.25em]e2q.west);
    %
    \draw[->](suml1.east) -- (G_E.west);
    \draw[->]([xshift = -0.75 em]suml1.west) |- ([yshift = 0.5em]tc.west);
    \draw[->]([xshift = -5.5em, yshift = -0.5em]tc.west) -- node[above, xshift = -2em]{ {}$V_{\rmm}$} ([yshift = -0.5em]tc.west);
    \draw[->]([xshift = -3.25em]suml1.west) -- node[above, xshift = -1em]{ {}$\begin{bmatrix}\Phi_{\rms} \\ \Theta_{\rms}\end{bmatrix}$} (suml1.west);
    \draw[->]([yshift = 1em]suml1.north) -- node[xshift = -1.4 em, yshift = 1.75em]{ {}$\begin{bmatrix}\Phi_{\rmm} \\ \Theta_{\rmm} \end{bmatrix}$} node[xshift = -0.15em, yshift = -0.05em]{$-$} (suml1.north);
    \draw[->](e2q.east)--node[above, xshift = -0.4em]{ {}$\omega_{\rms}$}(sumr1.west);
    \draw[->](sumr1.east) -- (G_PI.west);
    \draw[->]([yshift = -1em]sumr1.south) -- node[xshift = 1em, yshift = -1em]{ {}$\omega_{\rmm}$} node [xshift = 1.25em, yshift = -0.1em]{$-$} (sumr1.south);
    \draw[->]([xshift = 1.75em]e2q.east) |- (G_FF.west);
    \draw[->](G_FF.east)-|
    (sc.north);
    \draw[->](G_PI.east) -- 
    (sc.west);
    \draw[->]([yshift = -1em]sc.south) -- node[xshift = 1.75em, yshift = -1.25em]{ {}$V_{\rmT}, V_{\rmI}$} (sc.south);
    \end{tikzpicture}
    }
    \caption{\footnotesize PX4 autopilot attitude controller.}
    \label{fig:PX4_autopilot_inner_loop}
\end{figure}
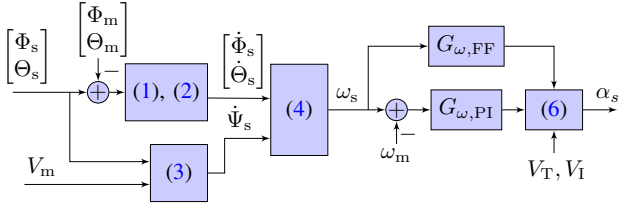

The fixed-wing autopilot thus consists of 11 gains.
%
In practice, these 11 gains are manually tuned and require considerable expertise. 

\section{RCAC algorithm}
\label{sec:RCAC}


This section briefly reviews the retrospective cost adaptive control (RCAC) algorithm. 
%
Consider a SISO PID controller with a feedforward term
\begin{align}
    u_k
        &=
            K_{\rmp,k} z_{k-1} +
            K_{\rmi,k} \gamma_{k-1} 
            \nn \\ &\quad +
            K_{\rmd,k} (z_{k-1} - z_{k-2}) +
            K_{{\rm ff},k} r_k
            ,
    \label{eq:uk_PID}
\end{align}
where $K_{\rmp,k}, K_{\rmi,k}, K_{\rmd,k}, $ and $K_{{\rm ff},k}$ are time-varying gains to be optimized, 
$z_k$ is an error variable,
$r_k$ is the feedforward signal, 
and, for all $k\ge0$,
\begin{align}
    \gamma_k 
        \isdef
            \sum_{i=0}^{k} g(z_{i}).
\end{align}
Note that the integrator state is computed recursively using $\gamma_k = \gamma_{k-1} + g(z_{k-1})$.
For all $k\ge0$, 
the control law can be written as 
\begin{align}
    u_k 
        =
            \phi_k \theta_k,
    \label{eq:uk_reg}
\end{align}
where 
the regressor $\phi_k$ 
and 
the controller gains $\theta_k$ 
are 
\begin{align}
    \phi_k
        \isdef
            \matl{c}
                z_{k-1} \\
                \gamma_{k-1} \\
                z_{k-1} - z_{k-2} \\
                r_k
            \matr^\rmT, \quad
    \theta_k
        \isdef
            \matl{c}
                K_{\rmp,k} \\
                K_{\rmi,k} \\
                K_{\rmd,k} \\
                K_{{\rm ff},k}
            \matr .
    \label{eq:phi_theta_def}
\end{align}
Note that the P, PI, or PID controllers can be parameterized by appropriately defining $\phi_k$ and $\theta_k.$ 
Various MIMO controller parameterizations are shown in \cite{goel_2020_sparse_para}.

To determine the controller gains $\theta_k$, let $\theta \in \BBR^{l_\theta}$, and consider the \textit{retrospective performance variable} defined by
\begin{align}
    \hat{z}_{k}(\theta)
        \isdef
            z_k + 
            \sigma (\phi_{k-1} \theta - u_{k-1}),
    \label{eq:zhat_def}
\end{align}
where $\sigma \in \BBR.$
The sign of $\sigma$ is the sign of the leading numerator coefficient of the transfer function from $u_k$ to $z_k.$
Furthermore, define the \textit{retrospective cost function} $J_k \colon \BBR^{l_\theta} \to [0,\infty)$ by
\begin{align}
    J_k(\theta) 
        &\isdef
            \sum_{i=0}^k
                \hat{z}_{i}(\theta) ^\rmT 
                R_z
                \hat{z}_{i}(\theta)
                 +
                (\phi_k \theta)^\rmT
                R_u
                (\phi_k \theta)
                \nn \\ &\quad \quad +
                (\theta-\theta_0)^\rmT 
                P_0^{-1}
                (\theta-\theta_0),
    \label{eq:RetCost_def}
\end{align}
where $\theta_0\in\BBR^{l_\theta}$ is the initial vector of PID gains and $P_0\in\BBR^{l_\theta\times l_\theta}$ is positive definite.
%

\begin{proposition}
    Consider \eqref{eq:uk_reg}--\eqref{eq:RetCost_def}, 
    where $\theta_0 \in \BBR^{l_\theta}$ and $P_0 \in \BBR^{l_\theta \times l_\theta}$ is positive definite. 
    Furthermore, for all $k\ge0$, denote the minimizer of $J_k$ given by \eqref{eq:RetCost_def} by
    \begin{align}
        \theta_{k+1}
            \isdef
                \underset{ \theta \in \BBR^n  }{\operatorname{argmin}} \
                J_k({\theta}).
        \label{eq:theta_minimizer_def}
    \end{align}
    Then, for all $k\ge0$, $\theta_{k+1}$ is given by 
    \begin{align}
        \theta_{k+1} 
            &=
                \theta_k  - 
                 \sigma P_{k+1}\phi_{k-1}^\rmT R_z
                 [ z_k + \sigma(\phi_{k-1} \theta_k - u_{k-1}) ]
                 \\ &\quad \quad 
                 - 
                 P_{k+1}\phi_{k}^\rmT
                 R_u \phi_{k} \theta_k 
                 , \label{eq:theta_update}
    \end{align}
    where 
    \begin{align}
        P_{k+1} 
            &=
                P_{k}
                -  
                P_k  
            \Phi_k ^\rmT 
            \left( 
                \bar R ^{-1} +  
                \Phi_k
                P_k
                \Phi_k ^\rmT 
            \right)^{-1}
            \Phi_k,
        \label{eq:P_update_noInverse}
    \end{align}
    and
    \begin{align}
        \Phi_k
            \isdef 
            \matl{c}
                \sigma \phi_{k-1} \\
                \phi_k
            \matr,
        \quad 
        \bar R
            \isdef 
                \matl{cc}
                    R_z & 0 \\
                    0   & R_u
                \matr.
    \end{align}
\end{proposition}
\begin{proof}
See \cite{AseemRLS}.
\end{proof}

Finally, the control is given by
\begin{align}
    u_{k+1} = \phi_{k+1} \theta_{k+1}.
\end{align}

\section{Adaptive PX4 Autopilot}
\label{sec:adaptiveAugmentation}
This section describes the adaptive autopilot, which is constructed by augmenting the fixed-gain autopilot described in Section \ref{sec:PX4_autopilot}.
In particular, each fixed-gain control law is modified by adding an adaptive control law, which is updated by the retrospective cost adaptive control (RCAC) algorithm described in Section \ref{sec:RCAC}.
The output of the modified controller is thus given by the sum of the fixed-gain control law and the adaptive control law, as shown in Figure \ref{fig:Augmented_PX4_autopilot_inner_loop}. 

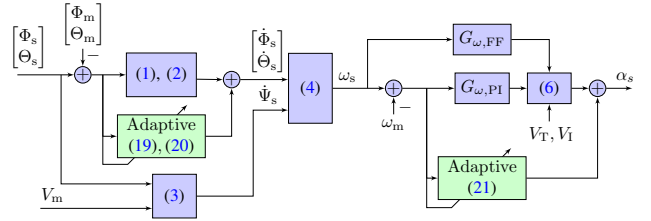
\begin{figure}[h]
    \centering
    \resizebox{\columnwidth}{!}{
    \begin{tikzpicture}[auto, node distance=2cm,>=latex']
    \node[smallblock, minimum width = 2.5em, minimum height = 4 em] (e2q) { \eqref{eq:e2omega}};
    \node[smallblock, above left = -0.5em and 6.5em of e2q.center, minimum width = 4em, minimum height = 2.5em] (G_E) {\eqref{eq:pitch_rate_P}, \eqref{eq:Roll_rate_P}};
    \node[smallblock, below left = 4.9em and 6.5em of e2q.center, minimum width = 2.5em, minimum height = 2.5em] (tc) {\eqref{eq:coordinated_turn}};
    \node[sum, left = 4em of G_E.center] (suml1){};
    \node[draw = none] at (suml1.center) {$+$};
    \node[sum, below right = -0.075em and 1.65em of G_E.east] (suml2){};
    \node[draw = none] at (suml2.center) {$+$};
    \node[sum, right = 3em of e2q] (sumr1){};
    \node[draw = none] at (sumr1.center) {$+$};
    \node [right = 0.3 em of suml1.east] (suml1_hook) {};
    \draw [->] (suml1_hook.center) -- +(0,-1.8) -- +(.5,-1.8) -- +(1.7,-.6);
    \node[smallblock, fill=green!20, below = 1em of G_E.south, minimum width = 0.5em, minimum height = 2 em, inner sep=0.25pt](G_E_adap){ {}$\begin{array}{c} {\rm Adaptive} \\ \eqref{eq:adaptive_pitch_rate_P}, \eqref{eq:adaptive_roll_rate_P} \end{array}$};
    \draw[->] (suml1_hook.center)|-(G_E_adap.west);
    \node[smallblock, right = 3em of sumr1, minimum width = 2.5em, minimum height = 1.75 em] (G_PI) {$G_{\omega, {\rm PI}}$};
    \node [left = 1.25 em of G_PI.west] (G_PI_hook) {};
    \draw [->] (G_PI_hook.center) -- +(0,-2.4) -- +(.5,-2.4) -- +(1.7,-1.2);
    \node[smallblock, fill=green!20, below = 3em of G_PI.south, minimum width = 0.5em, minimum height = 2 em, inner sep=0.25pt](G_PI_adap){ {}$\begin{array}{c} {\rm Adaptive} \\ \eqref{eq:alpha_s_adaptive} \end{array}$};
    \draw[->] (G_PI_hook.center)|-(G_PI_adap.west);
    \node[smallblock, above = 1em of G_PI, minimum width = 2.5em, minimum height = 1.75 em] (G_FF) {$G_{\omega, {\rm FF}}$};
    %
    %
    \node[smallblock, right = 1em of G_PI, minimum width = 2.5em, minimum height = 1.75 em] (sc) {\eqref{eq:alpha_s}};
    \node[sum, right = 1em of sc] (sumPI){};
    \node[draw = none] at (sumPI.center) {$+$};
    \draw [->] (sumPI.east) -- ([xshift = 1.5em]sumPI.east) node [above, xshift=-0.4 em] { {}$\alpha_s$};
    \draw[->]([yshift = -0.25em]G_E.east) -- (suml2.west);
    \draw[->](suml2.east) -- node[above]{ {} $\begin{bmatrix} \dot{\Phi}_{\rms} \\ \dot{\Theta}_{\rms} \end{bmatrix}$}([yshift = 0.5em]e2q.west);
    \draw[->](G_E_adap.east) -| (suml2.south);
    \draw[->](tc.east)-|([xshift = -2em, yshift = -1.25em]e2q.west)-- node[above, xshift = -0.15 em]{ {}$\dot{\Psi}_{\rms}$} ([yshift = -1.25em]e2q.west);
    \draw[->](suml1.east) -- (G_E.west);
    \draw[->]([xshift = -0.75 em]suml1.west) |- ([yshift = 0.7em]tc.west);
    \draw[->]([xshift = -6.4em, yshift = -0.7em]tc.west) -- node[above, xshift = -2.5em]{ {}$V_{\rmm}$} ([yshift = -0.7em]tc.west);
    \draw[->]([xshift = -3.25em]suml1.west) -- node[above, xshift = -1em]{ {}$\begin{bmatrix}\Phi_{\rms} \\ \Theta_{\rms}\end{bmatrix}$} (suml1.west);
    \draw[->]([yshift = 1em]suml1.north) -- node[xshift = -1.6 em, yshift = 1.9em]{ {}$\begin{bmatrix}\Phi_{\rmm} \\ \Theta_{\rmm} \end{bmatrix}$} node[xshift = -0.15em, yshift = 0.1em]{$-$} (suml1.north);
    \draw[->](e2q.east)--node[above, xshift = -0.65em]{ {}$\omega_{\rms}$}(sumr1.west);
    \draw[->](sumr1.east) -- (G_PI.west);
    \draw[->]([yshift = -1em]sumr1.south) -- node[xshift = 1em, yshift = -1.1em]{ {}$\omega_{\rmm}$} node [xshift = 1.4em, yshift = -0.1em]{$-$} (sumr1.south);
    \draw[->]([xshift = 2em]e2q.east) |- (G_FF.west);
    %
    %
    \draw[->](G_FF.east)-|
    (sc.north);
    \draw[->](G_PI_adap)-|(sumPI.south);
    \draw[->](G_PI.east)--(sc.west);
    \draw[->](sc.east) -- 
    (sumPI.west);
    \draw[->]([yshift = -1em]sc.south) -- node[xshift = 1.75em, yshift = -1.25em]{ {}$V_{\rmT}, V_{\rmI}$} (sc.south);
    \end{tikzpicture}
    }
    \caption{\footnotesize Adaptive PX4 autopilot attitude controller.}
    \label{fig:Augmented_PX4_autopilot_inner_loop}
\end{figure}
The bank and elevation rate setpoints $\dot{\Phi}_{\rms}, \dot{\Theta}_{\rms}$ in the adaptive autopilot are given by
%
\begin{align}
    \dot{\Theta}_{\rms} 
        &=
            k_\theta(\Theta_{\rms} - \Theta_{\rmm})
            + 
            u_\Theta
        \label{eq:adaptive_pitch_rate_P}
        , \\        
    \dot{\Phi}_{\rms} 
        &=
            k_\phi (\Phi_{\rms} - \Phi_{\rmm})
            +
            u_\Phi
            , 
        \label{eq:adaptive_roll_rate_P}
\end{align}
where $u_\Theta$ and $u_\Phi$ are computed by the RCAC algorithm.
Note that $u_\Theta$ and $u_\Phi$ are scalars. 
Similarly, the angular acceleration setpoint in the adaptive autopilot is given by 
\small
\begin{align}
    \alpha_s
        &=
            \frac{V_{\rmT, 0}}{V_{\rmT}} G_{\omega,{\rm ff}} 
            \omega_\rms 
            \nn \\ 
            & \quad +
            \left(\frac{V_{\rmI,0}}{V_{\rmI}}\right)^2 G_{\omega,{\rm PI}} (\shiftq) 
            \left( \omega_\rms - \omega_\rmm
            \right)
            +
            u_{\omega, \rm PI}, \label{eq:alpha_s_adaptive}
\end{align}
\normalsize
where
$u_{\omega, \rm PI}$ is computed by the RCAC algorithm.
Note that $ u_{\omega, \rm PI} \in \BBR^3,$ and each component of $u_{\omega, \rm PI}$ is updated by RCAC, where the error variable is the corresponding error term.

\section{Simulation Results}
\label{sec:flight_tests}

This section investigates the performance of the adaptive autopilot and compares it to the performance of the fixed-gain autopilot with stock PX4 gains. 
The fixed-wing aircraft is simulated in Gazebo.
In particular, the standard plane model \footnote{\href{https://docs.px4.io/master/en/simulation/gazebo_vehicles.html\#standard_plane}{Gazebo standard plane simulator}} provided by the PX4 developers is used.
The default controller gains and the actuator constraints in PX4 are specified in the
\verb|fw_att_control_params.c|
\footnote{\href{https://github.com/JAParedes/PX4-Autopilot/tree/RCAC_FW_dev}{https://github.com/JAParedes/PX4-Autopilot/tree/RCAC\_FW\_dev}}.
Note that the results in this paper are based on PX4 version V1.10.0dev.

The default controller gains in PX4 are well-tuned for the Gazebo model considered in this paper. 
To investigate the potential improvements in the performance, the default controller gains in PX4 are multiplied by a scalar $\alpha_\rmd$ to degrade the performance of the autopilot. 
This is equivalent to the case of poor initial choice of controller gains. 
The baseline performance is obtained with $\alpha_\rmd=1$.
A value of $\alpha_\rmd \neq 1$ is assumed to degrade the performance of the autopilot.
%
%
In the case where $\alpha_\rmd = 0,$ the fixed-gain autopilot is completely switched off, and the adaptive autopilot learns the controller gains in-situ.
%
%
Table \ref{tab:RCPE_variables_SITL} lists the hyperparameters used by RCAC.
Note that the RCAC hyperparameters are not changed as $\alpha_\rmd$ is varied. 

Figure \ref{fig:flightPathFW} shows the trajectory of the fixed-wing aircraft for various values of $\alpha_\rmd$ with both the fixed-gain and the adaptive autopilot. 
The baseline performance is obtained with $\alpha_\rmd=1,$ and in this case, adaptive autopilot does not affect the performance. 
For $\alpha_\rmd=0.5,$ the trajectory following response degrades substantially for the fixed-gain autopilot, and in this case, the adaptive autopilot recovers the baseline performance. 
Finally, the adaptive autopilot is also able to learn the gains in situ, in the case where the fixed-gain autopilot is completely switched off, that is, $\alpha_\rmd=0.$
Figure \ref{fig:thetaPlot} shows the gains adapted by RCAC for the three values of $\alpha_\rmd. $
Note that as the fixed-gain controller is detuned, RCAC compensates by providing larger values of gains. 


\begin{table}[h]
    \caption{\footnotesize RCAC hyperparameters in the adaptive autopilot. 
    }
    
    \label{tab:RCPE_variables_SITL}
    \centering
    \renewcommand{\arraystretch}{1.2}
    \begin{tabular}{|c|l|l|}
        \hline
        \multicolumn{1}{|c|}{\textbf{Controller}}  & \multicolumn{1}{|c|}{${P_0}$}  & \multicolumn{1}{|c|}{${R_u}$}
        \\ \hhline{|=|=|=|}
        \eqref{eq:adaptive_pitch_rate_P}, $\theta_\Theta$
        & $0.01$ & $0.001$
        \\ \hline
        \eqref{eq:adaptive_roll_rate_P}, $\theta_\Phi$ & $1$ & $0.001$
        \\ \hline
        \eqref{eq:alpha_s_adaptive}, $\theta_{\dot \Theta}$
        & $1000$ & $0.1$
        \\ \hline
        \eqref{eq:alpha_s_adaptive}, $\theta_{\dot \Phi}$
        & $0.001$ & $0.1$
        \\ \hline
    \end{tabular}
\end{table}

\begin{figure}
    \centering
    \includegraphics[width=\columnwidth]{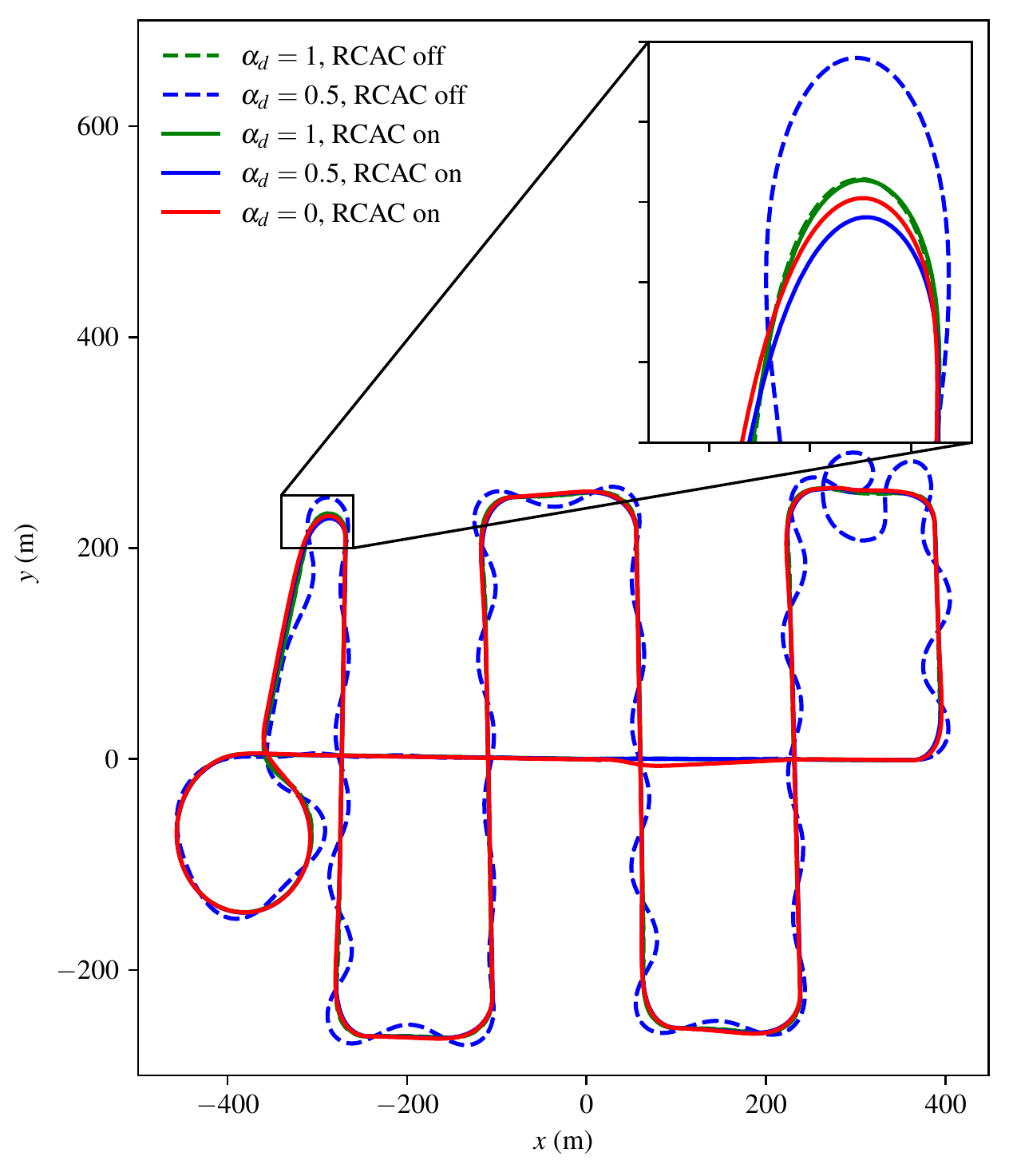}
    \caption{Trajectory-following response of the fixed-wing aircraft with the fixed-gain and the adaptive autopilot for various values of $\alpha_\rmd.$
    }
    \label{fig:flightPathFW}
\end{figure}

\begin{figure}
    \centering
    \includegraphics[width=\columnwidth]{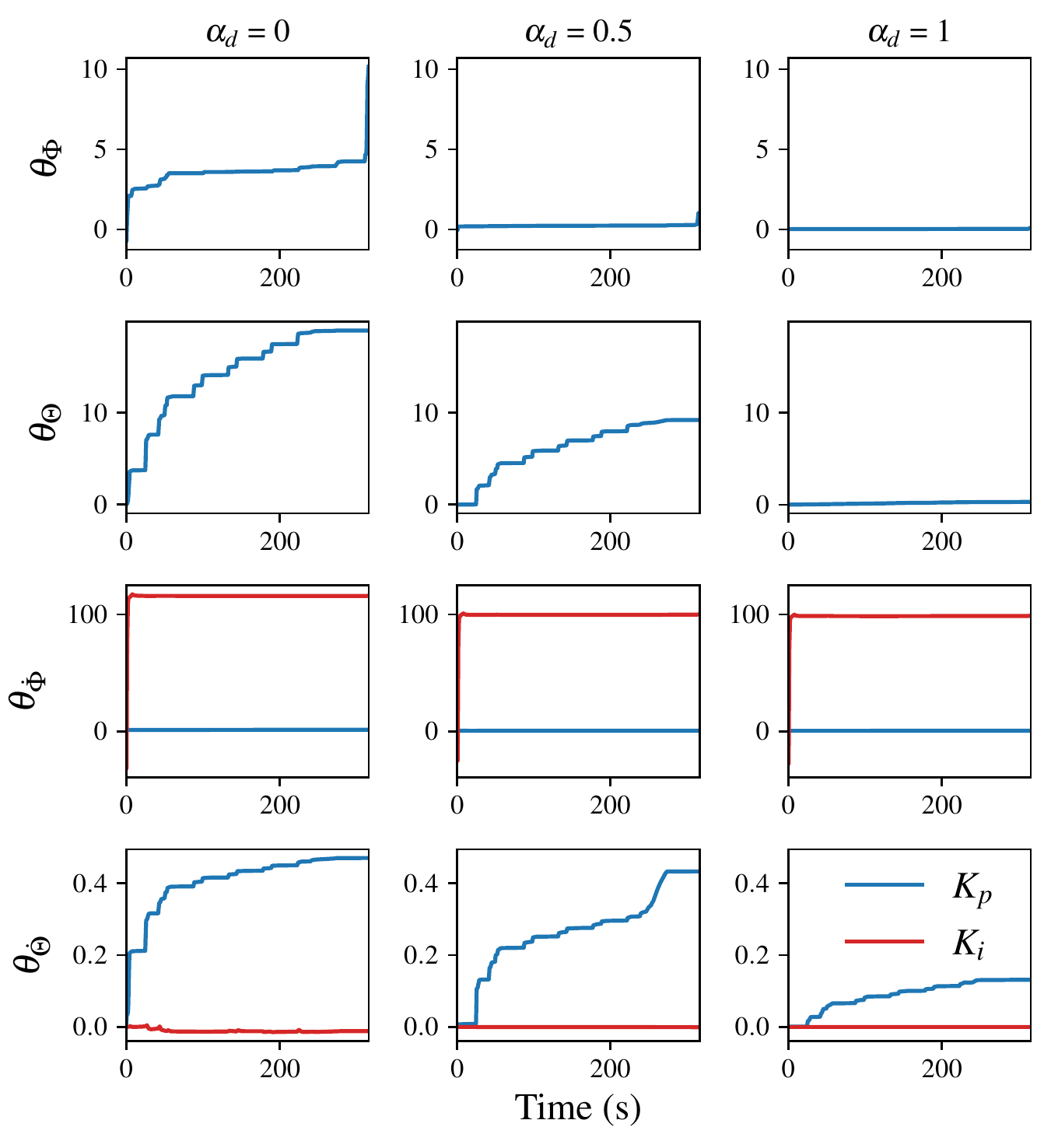}
    \caption{Adaptive autopilot gains adapted by the RCAC algorithm for various values of $\alpha_\rmd$.
    Note that as the fixed-gain controller is detuned, RCAC compensates by providing larger values of gains.
    In the case where $\alpha_\rmd=1,$ RCAC gains are small, whereas, for $\alpha_\rmd<1,$ RCAC gains are larger. 
    }
    \label{fig:thetaPlot}
\end{figure}

In order to quantify and compare the performance of the {adaptive
autopilot} with the {fixed-gain autopilot}, bank, elevation and trajectory-tracking error metrics are defined as
\begin{align}
    J_{\Phi} &\isdef \sqrt{\frac{1}{N} \sum_{i=1}^N (\Phi_{{\rms},i} - \Phi_{{\rmm},i})^2}, \label{eq:cost_bank} \\
    J_{\Theta} &\isdef \sqrt{\frac{1}{N} \sum_{i=1}^N (\Theta_{{\rms},i} - \Theta_{{\rmm},i})^2},\label{eq:cost_elevation} \\
    J_{\rm traj} &\isdef \sqrt{\frac{1}{N} \sum_{i=1}^N e_{{\rm x-track},i}^2},\label{eq:cost_xtrack}
\end{align}
where 
$N$ is the number of measurements during the flight, 
$e_{{\rm x-track}}$ is the cross-track error, defined as the minimum distance between the current position and desired trajectory. 
%
Note that these error metrics are computed after each flight test. 
Figure \ref{fig:RollCost} shows the error metrics for various values of $\alpha_\rmd$ with both the fixed-gain and the adaptive autopilot.

%

%

\begin{figure}[b]
\vspace{-0.5em}
    \centering
    \includegraphics[width=\columnwidth]{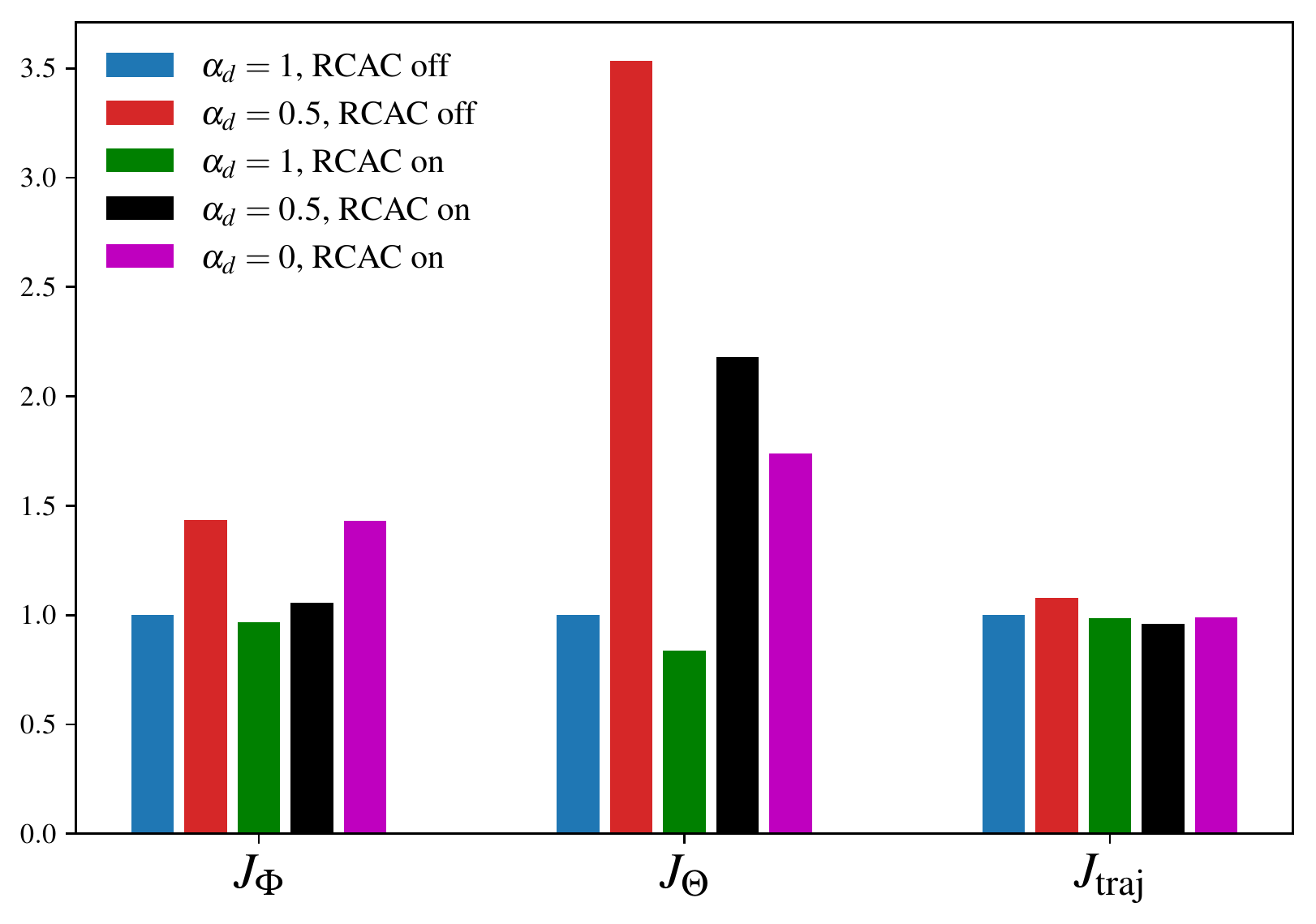}
    \caption{Bank, Elevation, and trajectory-tracking error metrics obtained with the fixed-gain and the adaptive autopilots for different values of $\alpha_\rmd.$
    Note that all of the error metrics are scaled by the benchmark error metric, where the benchmark flight is obtained by setting $\alpha_\rmd=1$ and switching off the adaptive autopilot.
    }
    \label{fig:RollCost}
\end{figure}

Next, the case of actuator failure is considered. 
In particular, the left aileron is assumed to be stuck at an unknown angular position. 
Figure \ref{fig:stuck_actuator} shows the actuator failure scenario. 
\begin{figure}
    \centering
    \includegraphics[width=\columnwidth]{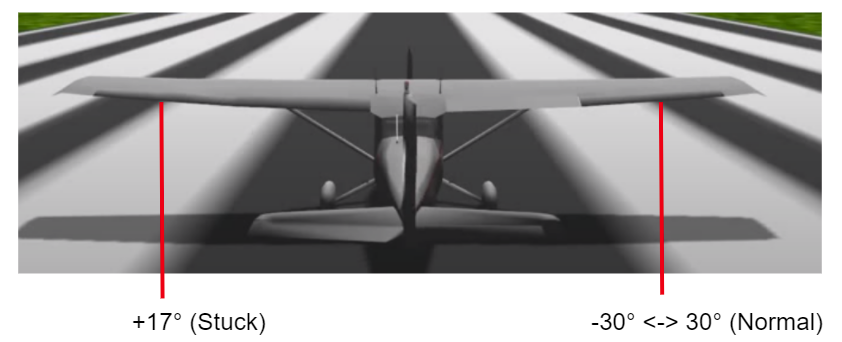}
    \caption{Actuator failure scenario. The left aileron is stuck at an unknown angular position. 
    }
    \label{fig:stuck_actuator}
\end{figure}

\begin{figure}
    \centering
    \includegraphics[width=\columnwidth]{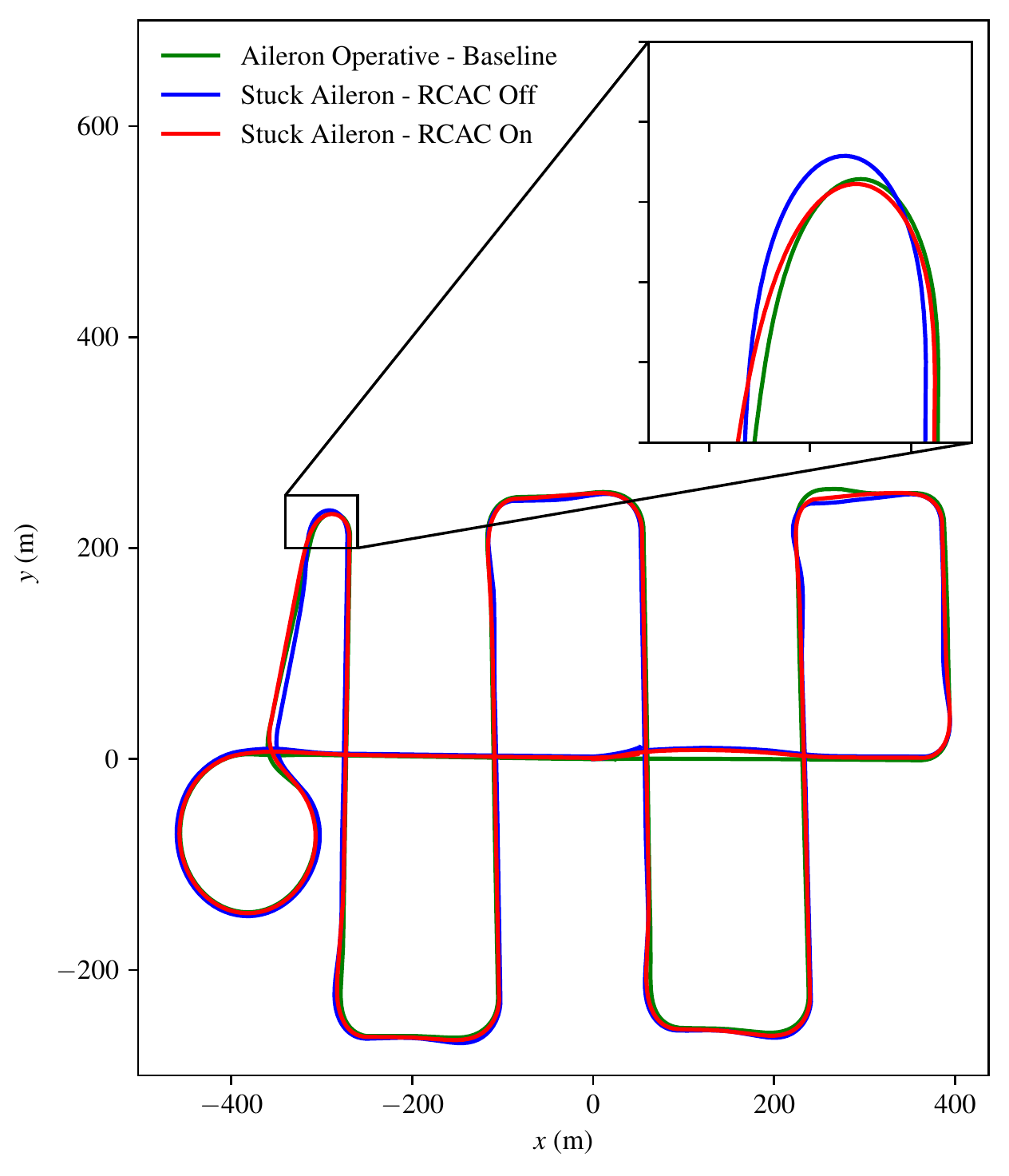}
    \caption{Trajectory-following response of the fixed-wing aircraft with frozen left aileron.
    }
    \label{fig:flightPathFW_stuck_actuator}
    \vspace{-1em}
\end{figure}

The aircraft is first commanded to follow a trajectory with the fixed-gain autopilot without any actuator failure.
This flight is considered to be the benchmark performance. 
Next, the left aileron is then frozen at an unknown angular position, and the aircraft is commanded to follow the same trajectory with the fixed-gain autopilot.
In this case, as expected, the trajectory-tracking error increases.
Finally, with the left aileron is frozen, the aircraft is commanded to follow the same trajectory with the adaptive autopilot.
Figure \ref{fig:flightPathFW_stuck_actuator} shows the trajectory-following response in the three cases. 
Note that the adaptive autopilot improves the trajectory-tracking error in the case of the stuck actuator and recovers the benchmark performance.
Figure \ref{fig:thetaPlot_stuck_actuator} shows the gains adapted by RCAC and Figure \ref{fig:RollCost_stuck_actuator} shows the error metrics in the case of the stuck actuator.

\begin{figure}
    \centering
    \includegraphics[width=\columnwidth]{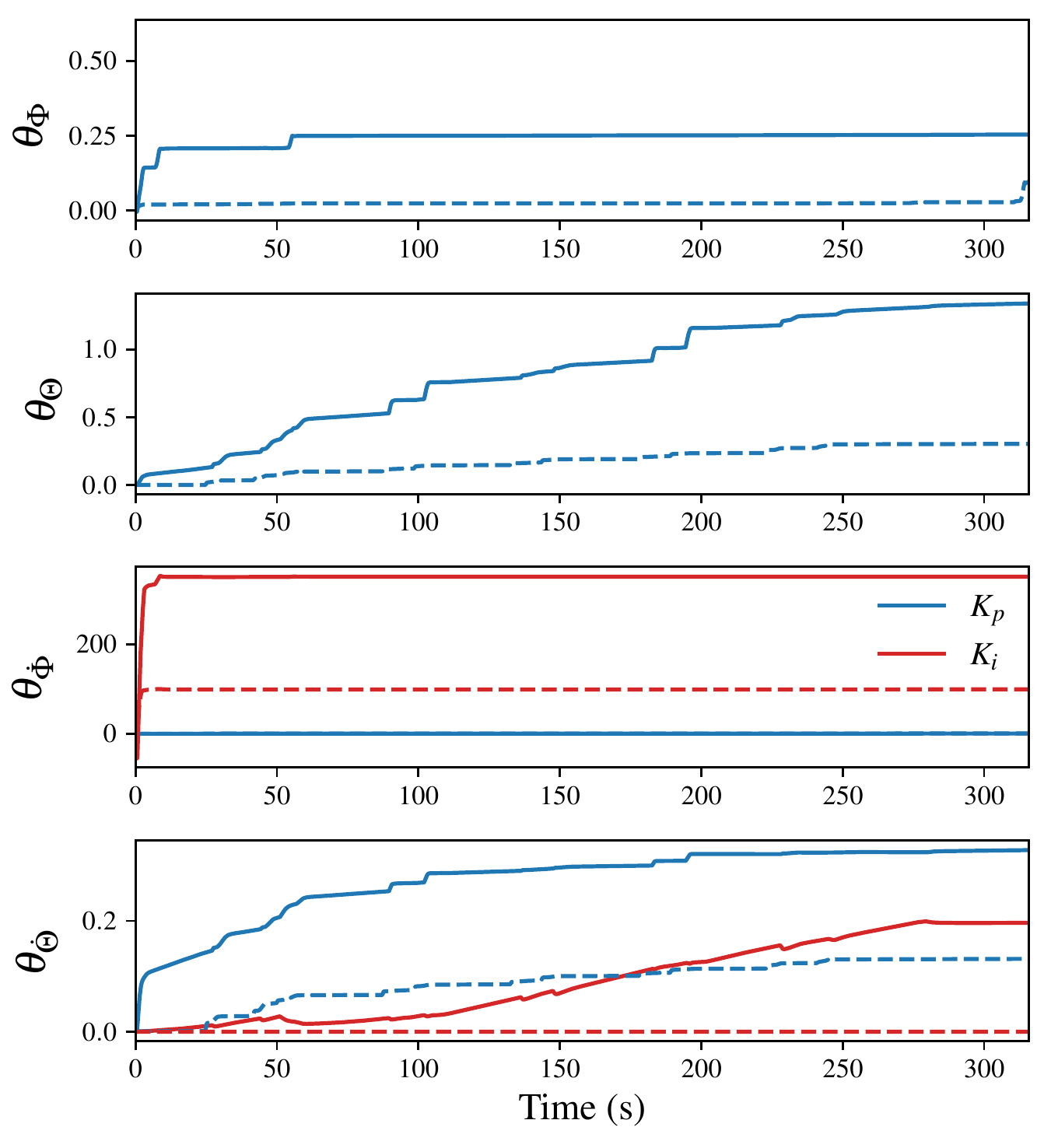}
    \caption{Adaptive autopilot gains adapted by the RCAC algorithm in the case of frozen aileron.
    Note that the dashed traces show the adaptive autopilot gains given by the RCAC algorithm in the nominal case, that is, with the operative aileron.
    }
    \label{fig:thetaPlot_stuck_actuator}
\end{figure}

\begin{figure}
    \centering
    \includegraphics[width=\columnwidth]{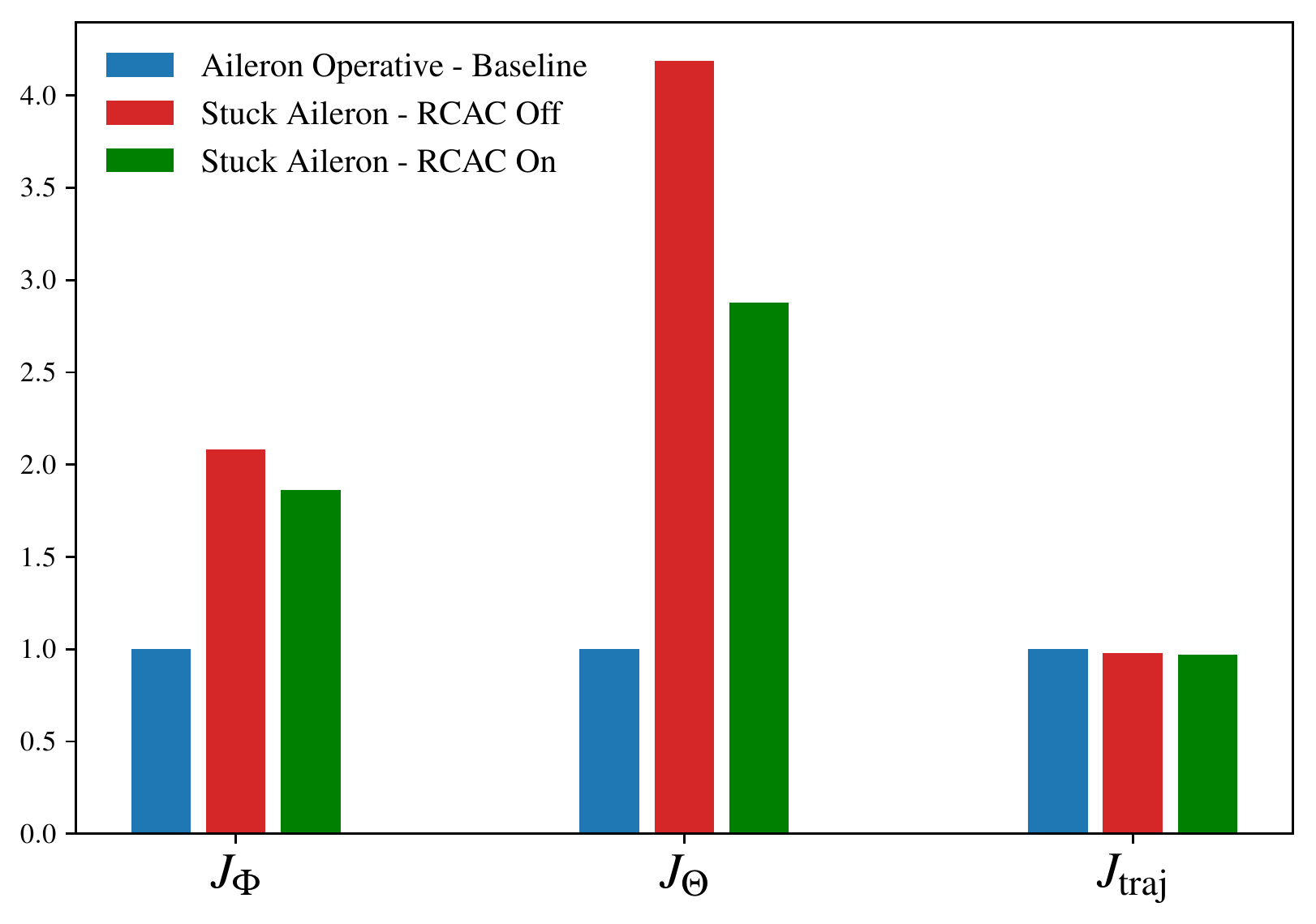}
    \caption{Bank, Elevation, and trajectory-tracking error metrics obtained with the fixed-gain and the adaptive autopilots for the actuator failure scenario.  
    Note that the error metrics are scaled by the benchmark error metric, where the benchmark flight is obtained with the operative aileron and without the adaptive autopilot.
    }
    \label{fig:RollCost_stuck_actuator}
\end{figure}

\section{Conclusions and Future Work}
\label{sec:conclusions}

%
This paper presented an adaptive digital autopilot that can improve an initial poor choice of controller gains in a fixed-gain autopilot, and learn the autopilot gains without any prior knowledge of the dynamics . 
The adaptive autopilot is constructed by augmenting the fixed-gain controllers in an autopilot with adaptive controllers.
%
The adaptive autopilot was used to fly a fixed-wing aircraft model in the Gazebo simulator.
The adaptive autopilot recovered the performance in the case where the fixed-gain autopilot was degraded and learned a set of gains in the case where the fixed-gain autopilot was completely switched off. 
Furthermore, the adaptive autopilot improved the trajectory-tracking performance in the case where the aileron was stuck at an unknown angular position.
Future work will focus on experimental investigation of the adaptive autopilot in order to assess the performance of the adaptive autopilot under sudden changes, such as stuck actuators, chipped wing, and changes in environmental conditions.

 \bibliographystyle{IEEEtran}
 \bibliography{main_FW_adaptive_arxiv}

\begin{thebibliography}{10}
\providecommand{\url}[1]{#1}
\csname url@samestyle\endcsname
\providecommand{\newblock}{\relax}
\providecommand{\bibinfo}[2]{#2}
\providecommand{\BIBentrySTDinterwordspacing}{\spaceskip=0pt\relax}
\providecommand{\BIBentryALTinterwordstretchfactor}{4}
\providecommand{\BIBentryALTinterwordspacing}{\spaceskip=\fontdimen2\font plus
\BIBentryALTinterwordstretchfactor\fontdimen3\font minus
  \fontdimen4\font\relax}
\providecommand{\BIBforeignlanguage}[2]{{%
\expandafter\ifx\csname l@#1\endcsname\relax
\typeout{** WARNING: IEEEtran.bst: No hyphenation pattern has been}%
\typeout{** loaded for the language `#1'. Using the pattern for}%
\typeout{** the default language instead.}%
\else
\language=\csname l@#1\endcsname
\fi
#2}}
\providecommand{\BIBdecl}{\relax}
\BIBdecl

\bibitem{mcclamroch2011steady}
N.~H. McClamroch, \emph{Steady aircraft flight and performance}.\hskip 1em plus
  0.5em minus 0.4em\relax Princeton University Press, 2011.

\bibitem{nguyen2006dynamics}
N.~Nguyen, K.~Krishnakumar, J.~Kaneshige, and P.~Nespeca, ``Dynamics and
  adaptive control for stability recovery of damaged asymmetric aircraft,'' in
  \emph{AIAA Guidance, navigation, and control Conference and Exhibit}, 2006,
  p. 6049.

\bibitem{9262225}
Z.~Yu, H.~Badihi, Y.~Zhang, Y.~Ma, B.~Jiang, and C.-Y. Su, ``Fractional-order
  sliding-mode fault-tolerant neural adaptive control of fixed-wing uav with
  prescribed tracking performance,'' in \emph{2020 2nd International Conference
  on Industrial Artificial Intelligence (IAI)}, 2020, pp. 1--6.

\bibitem{9476716}
M.~Fu, Z.~Yu, and Y.~Zhang, ``Adaptive fault-tolerant control of fixed-wing uav
  under actuator saturation and state constraints,'' in \emph{2021
  International Conference on Unmanned Aircraft Systems (ICUAS)}, 2021, pp.
  47--52.

\bibitem{hirano2019controller}
S.~Hirano, K.~Uchiyama, and K.~Masuda, ``Controller design using backstepping
  algorithm for fixed-wing uav with thrust vectoring system,'' in \emph{2019
  International Conference on Unmanned Aircraft Systems (ICUAS)}.\hskip 1em
  plus 0.5em minus 0.4em\relax IEEE, 2019, pp. 1084--1088.

\bibitem{9189264}
J.~Xiong, Y.~Yang, Z.~Cheng, L.~Liu, Y.~Wang, and H.~Fan, ``Observer-like model
  reference adaptive augmenting based fixed-wing uav control,'' in \emph{2020
  39th Chinese Control Conference (CCC)}, 2020, pp. 6804--6809.

\bibitem{rahmanCSM2017}
Y.~Rahman, A.~Xie, and D.~S. Bernstein, ``{Retrospective Cost Adaptive Control:
  Pole Placement, Frequency Response, and Connections with LQG Control},''
  \emph{IEEE Control System Magazine}, vol.~37, no.~5, pp. 28--69, 2017.

\bibitem{rezaPID}
M.~Kamaldar, S.~A.~U. Islam, S.~Sanjeevini, A.~Goel, J.~B. Hoagg, and D.~S.
  Bernstein, ``{Adaptive digital PID control of first-order-lag-plus-dead-time
  dynamics with sensor, actuator, and feedback nonlinearities},''
  \emph{Advanced Control for Applications}, vol.~1, no.~1, p. e20, 2019.

\bibitem{goel_adaptive_pid_2021}
A.~Goel, J.~A. Paredes, H.~Dadhaniya, S.~A. Ul~Islam, A.~M. Salim, S.~Ravela,
  and D.~Bernstein, ``Experimental implementation of an adaptive digital
  autopilot,'' in \emph{2021 American Control Conference (ACC)}, 2021, pp.
  3737--3742.

\bibitem{quadtuner2021}
J.~Spencer, J.~Lee, J.~A. Paredes, A.~Goel, and D.~Bernstein, ``An adaptive
  {PID} autotuner for multicopters with experimental results,''
  \emph{arXiv:2109.12797}, 2021.

\bibitem{bruce1986}
K.~Bruce, J.~Kelly, and J.~Person, ``{NASA B737} flight test results of the
  total energy control system,'' in \emph{Astrodynamics Conference}, 1986, p.
  2143.

\bibitem{faleiro1999}
L.~Faleiro and A.~Lambregts, ``Analysis and tuning of a {Total Energy Control
  System} control law using eigenstructure assignment,'' \emph{Aerospace
  science and technology}, vol.~3, no.~3, pp. 127--140, 1999.

\bibitem{lambregts2013}
A.~A. Lambregts, ``{TECS} generalized airplane control system design--an
  update,'' in \emph{Advances in Aerospace Guidance, Navigation and
  Control}.\hskip 1em plus 0.5em minus 0.4em\relax Springer, 2013, pp.
  503--534.

\bibitem{Argyle2016}
M.~E. Argyle and R.~W. Beard, ``Nonlinear total energy control for the
  longitudinal dynamics of an aircraft,'' in \emph{2016 American Control
  Conference (ACC)}, 2016, pp. 6741--6746.

\bibitem{park2004}
S.~Park, J.~Deyst, and J.~How, ``A new nonlinear guidance logic for trajectory
  tracking,'' in \emph{AIAA guidance, navigation, and control conference and
  exhibit}, 2004, p. 4900.

\bibitem{goel_2020_sparse_para}
A.~{Goel}, S.~A. {U. Islam}, and D.~S. {Bernstein}, ``{Adaptive Control of MIMO
  Systems Using Sparsely Parameterized Controllers},'' in \emph{2020 American
  Control Conference (ACC)}, 7 2020, pp. 5340--5345.

\bibitem{AseemRLS}
S.~A.~U. {Islam} and D.~S. {Bernstein}, ``{Recursive Least Squares for
  Real-Time Implementation},'' \emph{IEEE Control Systems Magazine}, vol.~39,
  no.~3, pp. 82--85, 6 2019.

\end{thebibliography}

\end{document}